\documentclass[conference]{IEEEtran}
\IEEEoverridecommandlockouts

\usepackage{quantikz}
\usepackage{amsmath, amssymb, amsthm, amsfonts}
\usepackage[ruled,vlined]{algorithm2e}
\usepackage[english]{babel}
\usepackage{listings}
\usepackage{tikz}
\usetikzlibrary{arrows.meta,positioning}
\theoremstyle{definition}
\newtheorem{definition}{Definition}[section]
\theoremstyle{remark}
\newtheorem{remark}{Remark}[section]
\newtheorem{example}{Example}[section]
\newtheorem{theorem}{Theorem}[section]
\usepackage{subcaption}

\usepackage{pgfplots}
\usepackage{graphicx}
\usepackage{url}
\usepackage{hyperref}
\usepackage{orcidlink}
\usepackage[capitalize]{cleveref}

\crefname{paragraph}{Paragraph}{Paragraphs}
\Crefname{paragraph}{Paragraph}{Paragraphs}

\newcommand{\ZERO}{\mathsf{ZERO}}
\newcommand{\ONE}{\mathsf{ONE}}
\newcommand{\PROB}{\mathsf{PROB}}

\tikzset{
    mybox/.style={
        draw,
        rounded corners,
        very thick,
        minimum width=2cm,
        minimum height=2cm,
        align=center,
        top color=#1!10,
        bottom color=#1!50
    }
}

\lstdefinelanguage{Pseudocode}{
  morekeywords={if,then,else,endif,while,do,endwhile,for,to,endfor,return,procedure,begin,end},
  sensitive=false,
  morecomment=[l]{//},
  morestring=[b]"
}

\usepackage{cite}
\usepackage{algorithmic}
\usepackage{graphicx}
\usepackage{textcomp}

\usepackage{xcolor}
\def\BibTeX{{\rm B\kern-.05em{\sc i\kern-.025em b}\kern-.08em
    T\kern-.1667em\lower.7ex\hbox{E}\kern-.125emX}}
\pgfplotsset{compat=1.18}
\begin{document}
\SetAlFnt{\small}
\newcommand{\mycommfont}[1]{\footnotesize\color{gray} #1}
\SetCommentSty{mycommfont}

\title{Compile-Time Simplification of Classically Controlled Operations in Dynamic Circuits

}

\author{\IEEEauthorblockN{Innocenzo Fulginiti \orcidlink{0000-0001-8818-9626}, Yanbin Chen \orcidlink{0000-0002-1123-1432}, Christian B.~Mendl \orcidlink{0000-0002-6386-0230}, and Helmut Seidl \orcidlink{0000-0002-2135-1593}}

\IEEEauthorblockA{TUM School of CIT, Technical University of Munich, Garching, Germany\\
\{innocenzo.fulginiti, yanbin.chen, christian.mendl, helmut.seidl\}@tum.de}}


\maketitle

\begin{abstract}
Dynamic circuits use real-time outcomes of mid-circuit measurements, processed by a classical controller, to adapt subsequent operations during circuit execution. This additional flexibility over static circuits comes at a price. Mid-circuit measurements are typically slower and noisier than unitary gates. Furthermore, classical feedforward requires exchanging information between the quantum processor (QPU) and the classical controller, introducing latency that erodes the practical performance of dynamic circuits. We propose a compile-time optimization framework that reduces the use of classical controls in dynamic circuits while preserving their semantics. At its core, the framework uses a static analysis that symbolically executes the circuit by propagating classical information alongside the quantum state. By combining this classical–quantum information with the Probabilistic Circuit Model extended with probabilistic controls that emulate classical feedforward, we obtain an intermediate probabilistic representation of the dynamic circuit. In this representation, mid-circuit measurements and classically controlled operations can be removed or rewritten as purely unitary operations and probabilistic components. Compared to existing compile-time optimizations that target only mid-circuit measurements, our method applies to a broader class of dynamic circuits expressible in modern quantum programming languages.
We evaluated our framework on randomly generated dynamic circuits, achieving about $50\%$ classical feedforward reduction and even higher reductions in favorable settings.
\end{abstract}


\section{Introduction}\label{sec:intro}
Dynamic circuits enable decision-making during execution by conditioning subsequent operations on measurement outcomes obtained at runtime. They combine mid-circuit measurements with classical feedforward, allowing conditional quantum operations driven by classical logic \cite{cirqQuantumaiClassicalControl, ibmClassicalFeedforward}.
This classical–quantum interaction expands the space of implementable algorithms and supports techniques such as distributed quantum computing, error correction and mitigation, qubit reuse, and quantum teleportation \cite{dqc, qec_1, exploiting_dyn_circ, qubit_reuse_1, qubit_reuse_2, teleportation_Bennett_1895}.
Unlike static circuits, whose gate sequence is fixed before execution, dynamic circuits can alter the computation’s control flow from shot to shot as measurement-dependent conditions are evaluated. This makes their behavior closer to classical programs with branching and conditionals, and has motivated modern quantum programming languages to support explicit classical control constructs within quantum code \cite{ibmClassicalFeedforward, openqasmClassicalInstructions, cirqClassicallyControlledOperation}.
Despite their benefits, dynamic circuits come with significant trade-offs on current hardware. Mid-circuit measurements and resets are typically slower than unitary gates and can add substantial runtime overhead. Moreover, these operations often incur higher error rates due to readout noise and additional decoherence while waiting for outcomes, making reliable mid-circuit measurements and resets challenging on noisy quantum devices \cite{mcm_hw_1, mcm_hw_error}. Furthermore, classically controlled operations require a feedback loop between the QPU and the classical controller: outcomes must be routed out, processed, and translated into updated control signals, introducing latency.

Modern quantum programming languages let developers express quantum programs using classical-like constructs such as variables, functions, and control flow \cite{openqasmClassicalInstructions, openqasmOpenQASMLive}. Although this improves programmability, it can also encourage less optimized implementations, making it easier to introduce unnecessary or redundant classical controls in dynamic circuits.
Thus, it is highly desirable to reduce the reliance on classically controlled operations within dynamic circuits, as doing so can mitigate latency and control overheads.
The optimization pass proposed in \cite{big_prob} combines Quantum Constant Propagation (QCP) \cite{qcp}, a static-analysis technique that propagates information about the quantum state through the circuit, with the Probabilistic Circuit Model (PCM) \cite{prob, pcm}, an intermediate representation that assigns stochastic compile-time semantics to selected components, to reduce the number of mid-circuit measurements and resets in dynamic circuits. This QCP-PCM based pass focuses on simplifying measurements and resets, while classical controls and classically controlled operations are mostly left untouched.
This is because the algorithm does not track classical information during compile-time simulation. As a result, it does not simplify quantum operations that depend on classical conditions. Consequently, this QCP-PCM based pass supports only a narrow class of dynamic circuits, and does not accommodate the more advanced constructs offered by modern quantum programming languages.

In this work, we address these limitations by extending the framework to a broader class of dynamic circuits. Our approach propagates classical information during compile-time simulation, tracking the evolution of the classical register alongside the quantum state. We also extend the Probabilistic Circuit Model with probabilistic controls, a construct that, combined with the inferred classical information, captures the effect of classical feedforward within the PCM representation. This allows classically controlled operations to be removed or rewritten using purely unitary operations and probabilistic components, reducing reliance on classical control in dynamic circuits. The overall procedure consists of two phases (\cref{fig:compilation-diagram}). In Phase~\textit{I}, the input dynamic circuit $D$ is analyzed and symbolically simulated at compile time, producing a semantically equivalent probabilistic circuit $D'$. In this probabilistic representation, classically feedforward can be replaced using unitary operations, probabilistic gates, and probabilistic controls. The resulting circuit $D'$ serves as an intermediate representation: it captures the effect of classical feedforward via stochastic components, but it is not directly executable, since probabilistic gates and controls must be mapped to concrete quantum operations.
In Phase~\textit{II}, we generate executable instances $D''$ of the probabilistic circuit $D'$ by compiling its probabilistic constructs into standard quantum operations. Whenever $D'$ has to be executed, one such instance $D''$ is sampled from the family of circuits represented by $D'$. In other words, for every execution of the optimized probabilistic circuit, a new executable circuit instance is generated.
\tikzset{probbox/.style={dashed, dash pattern=on 4pt off 2pt}}
\begin{figure}
    \centering
    \resizebox{\columnwidth}{!}{
    \begin{tikzpicture}[node distance=4cm]

    \node[mybox=red] (input) {Input \\ circuit \\ $D$};
    \node[mybox=blue, probbox, right of=input] (probabilistic) {Probabilistic \\ circuit \\ $D'$};
    \node[mybox=purple,right of=probabilistic] (instance) {Executable \\ instances \\ $D''$};
    
    \draw[->,thick]  (input) -- (probabilistic)
    node[midway,above,sloped] {Phase \textit{I}};
    \draw[->,thick]  (probabilistic) -- (instance)
    node[midway,above,sloped] {Phase \textit{II}};
    
    
    \end{tikzpicture}
    }
        \caption{
        Compilation overview. Phase~\textit{I} transforms the input dynamic circuit $D$ into a semantically equivalent probabilistic circuit $D'$, where classical feedforward is represented via probabilistic constructs. Phase~\textit{II} instantiates $D'$ by compiling probabilistic constructs into standard operations, sampling an executable circuit $D''$ for each execution.
        }
        \label{fig:compilation-diagram}
\end{figure}
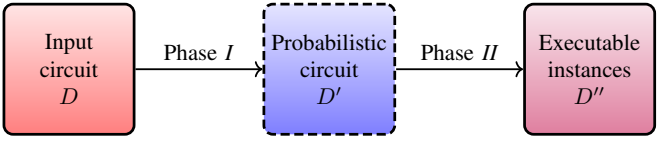

We implemented the proposed optimization procedure and evaluated it on a large dataset of randomly generated dynamic circuits to assess its effectiveness across circuit depths and widths. This experimental study quantifies the impact of our approach over a broad range of configurations. In addition, we present a case study on a dynamic circuit implementing an existing quantum algorithm from the literature.

\section{Preliminaries}\label{sec:prelim}
This work assumes that the reader is familiar with the fundamental concepts of quantum computing. For detailed background material, we refer the reader to textbooks such as \cite{nielsen_QC_2012, rieffel2000introduction, kaye2006introduction}.
Throughout this section, we fix notation and present the definitions that will be used in the rest of the paper.
\subsection{Quantum circuits and unitary operations}
A quantum circuit $D = (q,c)$ consists of a quantum register $q$, which is a sequence of qubits on which operations act, and a classical register $c$, which is a sequence of classical bits where outcomes of measurements performed on the qubits in $q$ are stored. If $q$ contains $n$ qubits we denote $\lvert q \rvert=n$, and $q_j$, for each $1 \leq j \leq n$, represents the $j$-th least significant qubit in $q$. In a similar way, if $c$ contains $m$ bits, we denote $\lvert c \rvert=m$, and $c_i$, for $1 \leq i \leq m$, is the $i$-th least significant bit in $c$.
\begin{remark}
    Given a quantum circuit $D = (q,c)$ we assume that each qubit $q_j$ is initialized in the quantum state $\ket{0}$ and each bit $c_i$ is initialized to $0$.
\end{remark}
\begin{definition}[Unitary Operation]
A unitary operation is the application of a $k$-qubit unitary operator $U$ to a tuple of qubits in the quantum register. We denote it by
$U(q_{r_1},\dots,q_{r_k})$ meaning that $U$ is applied to the qubits $q_{r_1},\dots,q_{r_k}$, where $1 \le r_\ell \le n$ for all $\ell \in \{1,\dots,k\}$ and the indices $r_1,\dots,r_k$ are pairwise distinct.
\end{definition}
\subsection{Dynamic circuits and classically controlled operations}
\label{subsec:dc_ccop}
Classically controlled operations are the essential building blocks of dynamic circuits. They allow the control flow of a quantum circuit to be modified at runtime based on classical information obtained by mid-circuit measurements.
\begin{definition}[Mid-circuit Measurement]
    A mid-circuit measurement is a measurement operation applied to a single qubit during circuit execution. It is denoted by $\mathrm{measure}(q_j, c_i)$,
meaning that qubit $q_j$ is measured in the computational basis $\{\ket{0},\ket{1}\}$ and the outcome is stored in classical bit $c_i$. Mid-circuit measurements are typically distinguished from final measurements, which are performed only at the end of the circuit. In this paper, the term measurement refers exclusively to mid-circuit measurement.
\end{definition}
\begin{definition}[Reset]
    A reset is a non-unitary operation that forces a qubit $q_j$ into the state $\ket{0}$, regardless of its prior state. We denote it by $\mathrm{reset}(q_j)$. Reset operations are used to reinitialize the state of a qubit (e.g., after a measurement).
\end{definition}
\begin{definition}[Classically controlled operation]
    A classically controlled operation consists of a classical condition on bits of a classical register, a \texttt{then} branch that is executed if the condition is satisfied, and an (optional) \texttt{else} branch that is executed if the condition is not satisfied. The syntactic structure of the classically controlled operations considered in this work is illustrated in \cref{fig:cls_ctrl_op}.
\end{definition}
\begin{remark}
In this paper, we restrict both the \texttt{then} and \texttt{else} branches to sequences of unitary operations only. In particular, branches do not contain mid-circuit measurements, resets, or further classically controlled operations.
\end{remark}
Tools such as OpenQASM 3.0, Cirq, and recent versions of Qiskit have extended support for expressing high-level conditions in a style closer to conventional programming languages for classical programs \cite{openqasmClassicalInstructions,cirqClassicallyControlledOperation,cirqQuantumaiClassicalControl,ibmClassicalFeedforward}. We can identify three syntactic patterns for writing conditions, namely single-bit comparison, full register comparison, and boolean predicate condition, that appear frequently in practice and are supported by the aforementioned languages.
\subsubsection{Single-bit comparison}
This widely used pattern involves conditions on a single bit of the form $c_i = 0$ or $c_i = 1$, where $c_i$ denotes the $i$-th bit of a classical register $c$.
\subsubsection{Full-register comparison}
In this pattern, a whole classical register $c$ of $m$ bits is compared with an integer value ranging from $0$ to $2^m -1$, in other words, a nonnegative integer of $m$ bits. (This is the style used in QASM 2.0 \cite{qasm2}).
\subsubsection{Boolean predicate condition}
This pattern employs boolean operators like \textit{and} $(\land)$, \textit{or} ($\lor$), \textit{xor} ($\oplus$) and \textit{not} $(\neg)$ to define boolean expressions of bits $c_i$ of a classical register. We consider boolean expressions generated by the rule
\begin{equation}
\label{eq:bool_exp}
E ::= c_i 
    \,\mid\, \neg E 
    \,\mid\, E_1 \land E_2 
    \,\mid\, E_1 \lor E_2 
    \,\mid\, E_1 \oplus E_2 \,.
\end{equation}
\begin{definition}[Dynamic Circuit]
    A dynamic circuit is a quantum circuit that interleaves unitary operations with mid-circuit measurements and classically controlled operations. Outcome values from mid-circuit measurements are stored in a classical register and are used to determine which operations are executed next, allowing branches within the computation.
\end{definition}

\begin{figure}
    \centering
\begin{lstlisting}
if cond then
    // A sequence of unitary operations executed if cond is satisfied
else
    // A sequence of unitary operations executed if cond is not satisfied
\end{lstlisting}
    \caption{Generic pattern for a classically controlled operation.}
    \label{fig:cls_ctrl_op}
\end{figure}
\subsection{Compile time and runtime for quantum circuits}

\begin{definition}[Compile time]
Compile time is the phase in which a quantum circuit is processed by classical software for compilation and optimization, producing an executable representation for a target backend.
\end{definition}

\begin{definition}[Runtime]
Runtime is the phase in which the compiled quantum circuit is executed on the quantum device.
\end{definition}

\subsection{Quantum Constant Propagation (QCP)}
\label{subsec:qcp}
Quantum Constant Propagation (QCP) is a static-analysis technique that was designed to exploit compile-time quantum constant information to simplify operations in quantum circuits \cite{qcp}. Assuming a fixed initial state (e.g., all qubits initialized to $\ket{0}$), QCP propagates static information about qubit states through the circuit. To avoid the exponential cost of full state simulation in the presence of entanglement, QCP tracks entangled states only up to a fixed threshold, which bounds the size of entanglement groups. When an entanglement group exceeds this threshold, QCP stops tracking it and assigns the unknown abstract state $\top$ to all qubits in the group. This restricted simulation yields a polynomial-time procedure in the circuit size. Indeed, for a circuit with $n$ qubits and $g$ gates, the analysis runs in $\mathcal{O}(n \cdot g)$ time.

In our framework, we build on QCP as the foundational static analysis to simulate the circuit’s quantum state and to track compile-time information within the quantum register.

\subsection{Probabilistic Circuit Model (PCM)}
The Probabilistic Circuit Model (PCM) provides an intermediate representation for quantum circuits that include probabilistic gates, i.e., components whose behavior is determined stochastically at compile time \cite{big_prob, prob, pcm}.
\begin{definition}[Probabilistic Gate]
Following the formulation in \cite{big_prob}, a probabilistic gate acting on $k$ qubits is written as
 \begin{equation}
     \mathfrak{P}[(S_1, p_1), \dots, (S_l, p_l)](q_{r_1}, \dots, q_{r_k})\,,
 \end{equation} where $0\le p_i\le 1$ are probability values, such that $\sum_{i=1}^{l} p_i = 1$. Each $S_i$ is a $k$-qubit unitary operator acting on the qubits $q_{r_1}, \dots, q_{r_k}$.
\end{definition}
\begin{example}
    The 2-qubit probabilistic gate
    $\mathfrak{P}[(I \otimes H, p_1), (SWAP, p_2), (X \otimes Z, p_3)](q_2, q_4)$
    is instantiated at compile time by sampling one of its alternatives: with probability $p_1$ it applies $I$ to $q_2$ and $H$ to $q_4$; with probability $p_2$ it applies $SWAP$ to $(q_2,q_4)$; and with probability $p_3$ it applies $X$ to $q_2$ and $Z$ to $q_4$.
\end{example}

\begin{definition}[Probabilistic Circuit]
A probabilistic circuit is a quantum circuit that supports probabilistic gates.
\end{definition}
For each execution of a probabilistic circuit, every probabilistic gate is instantiated by sampling from its associated distribution, yielding a concrete (non-probabilistic) gate sequence. Hence, each execution corresponds to a (possibly different) circuit derived from the same stochastic specification.

\subsection{State Preparation and Rotation Unitaries}
\label{subsec:state-prep-rotation}
We write an $n$-qubit state as
$
    \ket{\psi} = \sum_{b \in \{0,1\}^n} \alpha_b \ket{b},
$
where $b \in \{0,1\}^n$ is an $n$-bit string, $\ket{b}$ denotes the corresponding computational basis state, and $\sum_{b \in \{0,1\}^n} |\alpha_b|^2 = 1$.
\begin{definition}[Size of a state]
    \label{def:size-state}
    The size of $\ket{\psi} = \sum_{b \in \{0,1\}^n} \alpha_b \ket{b}$, denoted by $\lvert \ket{\psi} \rvert$, is the number of computational basis states $\ket{b}$ with nonzero amplitude, i.e., $\alpha_b \neq 0$.
\end{definition}
State preparation is a basic building block in many quantum algorithms and consists of a procedure that, given an $n$-qubit quantum state $\ket{\psi}$, allows one to construct a circuit $C_{\ket{\psi}}$ such that $C_{\ket{\psi}} \ket{0}^{\otimes n} = \ket{\psi}.$ In general, preparing an arbitrary $n$-qubit state may require a circuit whose depth and gate count scale as $\mathcal{O}(2^n)$ in the worst case, although different algorithms for state preparation have been proposed that achieve more efficient performance \cite{state_prep_1, state_prep_2, state_prep_qiskit}.

In what follows, we will treat state preparation as a primitive to construct rotation unitaries.

\begin{definition}[Rotation Unitary]
\label{def:rotation-gate}
Let $\ket{\psi}$ and $\ket{\phi}$ be two $n$-qubit states, and let $C_{\ket{\psi}}$ and $C_{\ket{\phi}}$ be state-preparation circuits for $\ket{\psi}$ and $\ket{\phi}$, respectively. We define the rotation unitary from $\ket{\psi}$ to $\ket{\phi}$ as the $n$-qubit circuit
\(
R_{\ket{\psi} \rightarrow \ket{\phi}} \coloneqq C_{\ket{\phi}} \cdot C_{\ket{\psi}}^{-1},
\)
so that $R_{\ket{\psi} \rightarrow \ket{\phi}} \ket{\psi} = \ket{\phi}$.
Here $C_{\ket{\psi}}^{-1}$ denotes the inverse unitary of $C_{\ket{\psi}}$, and $(\cdot)$ denotes sequential circuit concatenation.
\end{definition}

\subsection{Removing mid-circuit measurements with PCM}
\label{subsec:remove-mcm-pcm}
Leveraging its stochastic structure, PCM has been used to implement a compile-time simplification pass in which mid-circuit measurements are replaced by probabilistic gates \cite{prob, big_prob}. The central idea of this pass is to utilize the quantum information obtained from simulating the input circuit through QCP to reduce the number of measurement operations.
\begin{example}
\label{ex:mcm-rem-qcp-pcm}
    As a small illustrative example, consider the $3$-qubit quantum circuit $D = (q, c)$, defined as
\begin{lstlisting}
$D.X(q_1)$
$D.H(q_2)$
$D.CX(q_1, q_3)$
$D.\mathrm{measure}(q_2, c_2)$
\end{lstlisting}
Assuming the register ordering $\ket{q_3}\ket{q_2}\ket{q_1}$ and the initial state $\ket{0}\ket{0}\ket{0}$, QCP simulation shows that immediately before $D.\mathrm{measure}(q_2,c_2)$, the register $q$ is in the state $\ket{1}\ket{+}\ket{1}$. Measuring $q_2$ in the computational basis therefore returns $0$ or $1$ with probability $1/2$ each. Following \cite{big_prob}, this effect can be emulated by replacing the measurement with a rotation $D.R_{\ket{+}\rightarrow\ket{0}}(q_2)$, which maps $\ket{+}$ to $\ket{0}$, followed by the probabilistic gate $D.\mathfrak{P}[(I,0.5),(X,0.5)](q_2)$. Since the measurement outcome stored in $c_2$ is not used by subsequent classical controls, the measurement can be eliminated without changing the remaining computation; if outcomes are later consumed, the QCP-PCM pass cannot generally apply this rewrite because it does not track classical information at compile time.
\end{example}
\cref{ex:mcm-rem-qcp-pcm} illustrates the QCP-PCM construction in the simplest setting, where the measured qubit is not entangled with the rest of the register. When the measured qubit belongs to an entangled group, however, the measurement generally affects the entire group: measuring one subsystem can leave the remaining, unmeasured qubits in a mixed post-measurement state. The QCP-PCM method described in \cite{big_prob} can replace measurements on entangled qubits and correctly emulate the effect of the measurement on the state of the measured qubit, but it does not, in general, preserve the correct post-measurement state of the unmeasured qubits. Instead, the method collapses the entire entangled group into a single compatible computational basis state, rather than reproducing the intended mixture. Since a faithful account of this behavior is important for semantic correctness, in \cref{subsec:ref_rep_mcmres} we refine this construction, providing a consistent formulation.





\subsection{Limitations of QCP-PCM-Based Method}
\label{subsec:limits-qcp-pcm}
Dynamic circuits can measure qubits, store outcomes in a classical register, and use these values to guard subsequent classically controlled operations. As shown in \cref{subsec:remove-mcm-pcm}, the quantum-state update induced by a measurement (i.e., collapse according to the outcome distribution) can be emulated using rotations and probabilistic gates. However, probabilistic gates cannot reproduce the classical side effects of measurement, namely writing the outcome into the classical register, and they cannot natively represent classical computation over stored outcomes. Consequently, the method in \cite{big_prob} does not fully support general classically controlled operations.
Concretely, it can handle only components of the following restricted form:
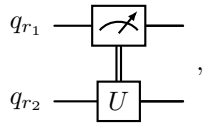
\begin{equation*}
    \centering
    \begin{quantikz}[row sep={10mm,between origins}]
          \lstick{$q_{r_1}$} & \meter{} \wire[d][1]{c} & \qw \\
            \lstick{$q_{r_2}$}   & \gate{U}                & \qw
    \end{quantikz} \,,
\end{equation*}
where the unitary $U$ is applied if and only if the measurement outcome of $q_{r_1}$ is $1$. Suppose the state of $q_{r_1}$ before the measurement is known to be $\ket{\psi}$ and $U$ acts on some target qubit $q_{r_2}$. The pass replaces the measurement by a rotation $R_{\ket{\psi}\rightarrow\ket{0}}(q_{r_1})$ followed by a probabilistic gate $\mathfrak{P}[(I,p_1),(X,p_2)](q_{r_1})$, and rewrites $U(q_{r_2})$ as its quantum-controlled version $CU(q_{r_1},q_{r_2})$, with $q_{r_1}$ as control qubit. In this way, $CU$ is applied exactly when the probabilistic gate is instantiated to $X$, emulating the original measurement-conditioned execution using unitary operations only.

This construction crucially relies on the measurement outcome being consumed immediately and once, as a direct guard for a single unitary, with no intermediate classical processing, storage, or reuse of the corresponding classical bit. As soon as outcomes are combined with other results, reused later in the circuit, or appear inside more general Boolean conditions, the QCP-PCM method has no mechanism to propagate and update the associated classical information. As a consequence, it does not extend directly to the broader class of classically controlled operations that arise in realistic dynamic circuits.

\section{Method}\label{sec:methods}
In this section, we present our optimization pass for simplifying classically controlled operations in dynamic circuits. We extend the Probabilistic Circuit Model with probabilistic controls (\cref{subsec:prob_controls}), which provide an abstract representation of classical controls. We then introduce Classical Constant Propagation (CCP) (\cref{subsec:ccp}), a static-analysis procedure that tracks classical information and is used alongside QCP during compile-time simulation. Using the information inferred by CCP, we derive rewriting rules to simplify the conditions of classically controlled operations (\cref{subsec:ccop-simp}). The overall optimization pass is presented in \cref{subsec:overall}, and its correctness is discussed in \cref{sec:method-correctness}.

\subsection{Refined rewriting of measurements and resets with PCM}
\label{subsec:ref_rep_mcmres}
We describe a construction for emulating the effect of measurements and resets on the quantum state of a circuit with PCM, which applies in the case where the measured qubit is entangled with other qubits, refining the procedures in \cite{big_prob}.
\subsubsection{Pattern for mid-circuit measurements}
\label{subsub:patt_mcm_emu}
Let $q_j$ be a qubit to be measured that is part of an $n$-qubit state $\ket{\psi}~=~\sum_{b \in \{0,1\}^n} \alpha_b \ket{b}$, over the qubits $\{q_1,\dots,q_n\}$. The state $\ket{\psi}$ can be decomposed with respect to $q_j$ as
$
\ket{\psi}
= \lambda_0 \ket{0}_{q_j}\ket{\phi_0} + \lambda_1 \ket{1}_{q_j}\ket{\phi_1},
$
where $|\lambda_0|^2 + |\lambda_1|^2=1$, and $\ket{\phi_0},\,\ket{\phi_1}$ are normalized $(n-1)$-qubit states representing, respectively, the state of the remaining qubits in the entanglement group when the measured qubit $q_j$ collapses to $\ket{0}$ or to $\ket{1}$ after the measurement. We define two rotation unitaries
\begin{equation}
\label{eq:rot-gate-in-prob-gate}
    \mathcal{R}_{\ket{0}} \coloneqq R_{\ket{0}^{\otimes (n-1)} \rightarrow \ket{\phi_0}},
\qquad
\mathcal{R}_{\ket{1}} \coloneqq R_{\ket{0}^{\otimes (n-1)} \rightarrow \ket{\phi_1}},
\end{equation}
such that $\mathcal{R}_{\ket{0}}\ket{0}^{\otimes (n-1)} = \ket{\phi_0}$ and $\mathcal{R}_{\ket{1}}\ket{0}^{\otimes (n-1)} = \ket{\phi_1}$.
As in the QCP-PCM construction of \cite{big_prob}, to replace the measurement, we first apply a rotation unitary $R_{\ket{\psi} \rightarrow \ket{0}^{\otimes n}}$ acting on all $n$ qubits, mapping $\ket{\psi}$ to $\ket{0}^{\otimes n}$. In \cite{big_prob}, this is followed by a probabilistic gate that samples one of the computational basis states compatible with $\ket{\psi}$, effectively replacing the post-measurement state with only a single computational basis state. Here, instead, given $\mathcal{R}_{\ket{0}},\,\mathcal{R}_{\ket{1}}$ (\cref{eq:rot-gate-in-prob-gate}), let $(q_{r_1},\dots,q_{r_{n-1}})$ be an arbitrary ordering of the qubits in $\{q_1,\dots,q_n\} \setminus \{q_j\}$, we define a probabilistic gate that correctly reconstructs the two post-measurement branches:
\begin{equation}
\label{eq:probgate_for_meas}
 \mathfrak{P}[(I \otimes \mathcal{R}_{\ket{0}}, |\lambda_0|^2), (X \otimes \mathcal{R}_{\ket{1}}, |\lambda_1|^2)](q_j, q_{r_1},\dots,q_{r_{n-1}}) \,.   
\end{equation}
With this construction, the probabilistic gate, when instantiated, reproduces one of the two post-measurement states of the $n$-qubit system initially in the state $\ket{\psi}$.

\subsubsection{Pattern for resets}
\label{subsub:patt_res_emu}
An analogous construction emulates resets. Suppose we reset qubit $q_j$ in the same $n$-qubit entangled state $\ket{\psi}$ as above.
To replace the reset, we again apply $R_{\ket{\psi} \rightarrow \ket{0}^{\otimes n}}$. In \cite{big_prob}, this step is followed by the assumption that the reset deterministically prepares the branch $\ket{0}\ket{\phi_0}$ without sampling. Here, instead, we preserve the correct probabilistic behavior by applying the probabilistic gate
\begin{equation}
    \label{eq:probgate_for_res} \mathfrak{P}[(I \otimes \mathcal{R}_{\ket{0}}, |\lambda_0|^2), (I \otimes \mathcal{R}_{\ket{1}}, |\lambda_1|^2)](q_j, q_{r_1},\dots,q_{r_{n-1}}) \,.
\end{equation}
In both branches, $q_j$ remains in state $\ket{0}$, while the remaining $n-1$ qubits are prepared in $\ket{\phi_0}$ or $\ket{\phi_1}$ with their probabilities, reproducing the post-reset behavior more accurately.
\subsubsection{Computational cost of rewriting patterns}
\label{subsec:comp-cost-rewriting}
    As analyzed in \cite{big_prob}, the dominant cost of these rewriting patterns is the synthesis of the rotation unitary: its complexity scales linearly with the number of computational-basis states in $\ket{\psi}$, and may therefore be exponential in the number of qubits in the worst case. For this reason, following \cite{big_prob}, we can impose a threshold on $\lvert \ket{\psi} \rvert$. If the threshold is exceeded, we do not replace the corresponding measurement or reset.

\subsection{PCM equipped with Probabilistic Controls}
\label{subsec:prob_controls}
In \cref{subsec:limits-qcp-pcm}, we discussed the simplified classically controlled operation case, where a single measurement outcome directly guards a unitary $U$. In that setting, the original QCP-PCM method replaces the measurement with a probabilistic gate and rewrites $U$ as a quantum-controlled operation $CU$. A naive extension to conditions depending on multiple classical bits would promote $U$ to a multi-controlled gate (e.g., $CCU$ for two qubits, and, more generally, a $k$-controlled gate for $k$ qubits). Such additional quantum controls increase gate cost and spread control dependencies across larger portions of the circuit.
More importantly, the original QCP-PCM abstraction supports only a restricted form of classical control: it effectively collapses the actual pattern of measuring a qubit, storing the outcome in a classical bit, and later using that value in a classical condition, into a single atomic component, thereby hiding the underlying classical control structure and preventing support for the richer conditional patterns illustrated in \cref{subsec:dc_ccop}. To capture these patterns explicitly, we introduce probabilistic controls, a new component of PCM.
\begin{definition}[Probabilistic control]
A probabilistic control is an abstract compile-time bit used in PCM to emulate the effect that a classical bit would have on the circuit control flow. A probabilistic control, denoted by $\mathfrak{C}_i$, takes values in $\{\mathfrak{0},\mathfrak{1}\}$, where $\mathfrak{C}_i = \mathfrak{0}$ and $\mathfrak{C}_i = \mathfrak{1}$ correspond, respectively, to the classical bit values $0$ and $1$ that would be stored in position $i$ of the classical register in the original dynamic circuit.
\end{definition}
\begin{definition}[Probabilistic gate-control binding]
    \label{def:gate-contr-bind}
    Consider a probabilistic control $\mathfrak{C}_i$, and a probabilistic gate $\mathfrak{P}[(I \otimes \mathcal{R}_{\ket{0}}, p_1), (X \otimes \mathcal{R}_{\ket{1}}, p_2)](q_j,\dots)$ as defined in \cref{eq:probgate_for_meas} for emulating mid-circuit measurements. We denote by
    $$
    \mathfrak{P}[(I \otimes \mathcal{R}_{\ket{0}}, p_1), (X \otimes \mathcal{R}_{\ket{1}}, p_2)](q_j,\dots)
    \rightarrow
    (\mathfrak{C}_i, q_j)
    $$
    the probabilistic gate-control binding between $\mathfrak{P}$ and $\mathfrak{C}_i$. The binding expresses a dependency between the sampling of $\mathfrak{P}$ and the value taken by $\mathfrak{C}_i$. When $\mathfrak{P}$ is instantiated, $\mathfrak{C}_i$ is set according to: 
    \begin{equation}
    \label{eq:prob-ctrl-update}
        \mathfrak{C}_i =
    \begin{cases}
        \mathfrak{0} & \text{if $I \otimes \mathcal{R}_{\ket{0}}$ is chosen},\\
        \mathfrak{1} & \text{if $X \otimes \mathcal{R}_{\ket{1}}$ is chosen}.
    \end{cases}
    \end{equation}

Thus, $\mathfrak{C}_i$ acts as a compile-time surrogate for the classical bit $c_i$, recording the value that the measurement of $q_j$ would write into the classical register.
\end{definition}
Probabilistic controls provide an abstract representation of measurement outcomes that would otherwise appear as explicit classical control signals, allowing us to emulate their effects when simplifying classically controlled operations. To exploit them systematically, we also track how classical information evolves throughout the circuit; this is performed by the Classical Constant Propagation analysis introduced in the next subsection.

\subsection{Classical Constant Propagation (CCP)}
\label{subsec:ccp}
Consider a dynamic circuit $D=(q,c)$ with $\lvert c \rvert = m$. The classical register $c$ is updated whenever a measurement writes its outcome into a classical bit $c_i$ of $c$; otherwise, its value remains unchanged.

We represent the abstract classical state of the register by associating to each bit $c_i$ a token from the finite set
\[
\mathsf{B} \coloneqq \{\ZERO,\ONE,\PROB,\top\}.
\]
Intuitively, $\ZERO$ and $\ONE$ indicate that the value of a bit $c_i$ is certainly $0$ and $1$, respectively, at a given circuit point. $\PROB$ indicates that $c_i$ is not fixed but behaves as a random bit whose value $0$ or $1$ follows a known probability distribution (e.g., $0$ and $1$ with probability $0.5$ after measuring a qubit in state $\ket{+}$). Finally, $\top$ represents a complete lack of information about the state of $c_i$.
We model the abstract state of $c$ as a function
\[
\sigma_c : \{1,\dots,m\} \to \mathsf{B},
\]
where $\sigma_c(i)$ is the token associated with the $i$-th bit of $c$.

At the beginning of the simulation, CCP initializes the abstract state according to the initial state of the circuit's classical register (e.g., all classical bits set to $0$, hence $\sigma_c(i) = \ZERO$ for all $i \in \{1,\dots,m\}$). As CCP simulates the dynamic circuit $D$, it updates $\sigma_c$ only when a measurement writes to the classical register. When a measurement $\mathrm{measure}(q_j, c_i)$ is encountered, CCP uses the information about the current quantum state of $q_j$ at that program point to update the token of the target bit $c_i$ accordingly. If it is inferred that the state of $q_j$ is deterministically $\ket{0}$ (respectively, $\ket{1}$), then $\sigma_c(i)$ is set to $\ZERO$ (respectively, $\ONE$). If the state of the qubit lies in a superposition, then $\sigma_c(i)$ is updated to $\PROB$. Finally, when no information on the state of $q_j$ is available, CCP sets $\sigma_c(i)$ to $\top$. All other entries $\sigma_c(k)$ with $k \neq i$ remain unchanged. In this way, CCP tracks the evolution of the classical information and exposes the current abstract state $\sigma_c$ at each point of the circuit. The effect of a measurement on the abstract classical state $\sigma_c$ is captured by the CCP update rule in \cref{alg:ccp}.
\begin{algorithm}[t]
\caption{CCP update for a measurement}
\label{alg:ccp}
\KwData{$\sigma_q$, abstract state of the quantum register $q$. $\sigma_c : \{1,\dots,m\} \to \mathsf{B}$, abstract state of the classical register $c$.}
\KwIn{A measurement instruction $\mathrm{measure}(q_j, c_i)$.}
\KwResult{Updated $\sigma_c$.}

\Switch{$\sigma_q(j)$}{
\uCase{$\ket{0}$}{$\sigma_c(i) \gets \ZERO$;}
\uCase{$\ket{1}$}{$\sigma_c(i) \gets \ONE$;}
\uCase{$\top$}{$\sigma_c(i) \gets \top$;}
\uCase{otherwise}{$\sigma_c(i) \gets \PROB$;}
}
\end{algorithm}

\subsection{Simplifying classically controlled operation conditions}
\label{subsec:ccop-simp}
We now define a static simplification function for the conditions of classically controlled operations, using the abstract classical state inferred for the dynamic circuit. The function applies to all condition patterns introduced in \cref{subsec:dc_ccop}.

\begin{definition}[Normal form for conditions]
\label{def:cond-norm-form}
Recall the grammar for Boolean predicate conditions from \cref{subsec:dc_ccop}. As an internal normal form used by our optimization pass, we represent all supported conditions as Boolean expressions generated by this grammar. In particular:
\begin{itemize}
    \item a single-bit comparison is translated as
    \[
        (c_i = 1) \;\equiv\; c_i,
        \qquad
        (c_i = 0) \;\equiv\; \neg c_i;
    \]
    \item a full-register comparison $c = val$, where $c$ is an $m$-bit register and $val \in \{0,\dots,2^m - 1\}$ has binary expansion $val_m \dots val_1$ with $val_i \in \{0,1\}$, is translated as
     \[
        (c = val) 
        \;\equiv\;
        \bigwedge_{i=1}^{m} E_i,
        \quad
        E_i \coloneqq 
        \begin{cases}
            c_i & \text{if } val_i = 1,\\
            \neg c_i & \text{if } val_i = 0.
        \end{cases}
    \]
\end{itemize}
We say that a condition is in normal form if it is expressed as a Boolean expression $E$ obtained by these translations.
\end{definition}

\begin{definition}[Probabilistic condition]
\label{def:prob-cond}
    A probabilistic condition is any Boolean expression over both classical and probabilistic controls generated by
$$
    E' \;::=\; c_i 
      \,\mid\, \mathfrak{C}_i
      \,\mid\, \neg E' 
      \,\mid\, E'_1 \land E'_2 
      \,\mid\, E'_1 \lor E'_2 
      \,\mid\, E'_1 \oplus E'_2,
$$
\end{definition}
Let $\mathcal{E}$ be the set of Boolean expressions $E$ (as in \cref{def:cond-norm-form}) and let $\mathcal{E}'$ be the set of probabilistic expressions $E'$ (as in \cref{def:prob-cond}). We define the abstract simplification function
\begin{equation}
\label{eq:simp-fun-cond}
\mathcal{S}_{\sigma_c} : \mathcal{E} \to \{true, false\} \cup \mathcal{E}' ,
\end{equation}
Intuitively, $\mathcal{S}_{\sigma_c}(E)$ simplifies $E$ recursively using the abstract information in $\sigma_c$. If $\mathcal{S}_{\sigma_c}(E)$ returns $true$ or $false$, the expression $E$ is statically evaluated at compile time. Otherwise, the output is a possibly simplified probabilistic condition $E' \in \mathcal{E}'$ in which: bits known to be constantly $\ZERO$ or $\ONE$ have been eliminated; bits in state $\PROB$ have been replaced by probabilistic controls $\mathfrak{C}_i$ (or their negations); bits in $\top$ remain as unresolved atoms $c_i$.

The definition of $\mathcal{S}_{\sigma_c}$ proceeds by structural induction on $E$. For atomic expressions, we have:
\[
    \mathcal{S}_{\sigma_c}(c_i) \;=\;
    \begin{cases}
        true & \text{if } \sigma_c(i) = \ONE,\\[2pt]
        false & \text{if } \sigma_c(i) = \ZERO,\\[2pt]
        \mathfrak{C}_i & \text{if } \sigma_c(i) = \PROB,\\[2pt]
        c_i & \text{if } \sigma_c(i) = \top.
    \end{cases}
\]
For the unary negation operator, we define:
\[
    \mathcal{S}_{\sigma_c}(\neg E) =
    \begin{cases}
        false & \text{if } \mathcal{S}_{\sigma_c}(E) = true,\\[2pt]
        true & \text{if } \mathcal{S}_{\sigma_c}(E) = false,\\[2pt]
        \neg E' & \text{if } \mathcal{S}_{\sigma_c}(E) = E' \in \mathcal{E}'.
    \end{cases}
\]

For the binary operators $(\land)$ and $(\lor)$ we use short-circuit identities and constant folding \cite{short-circ-eval}. For the $(\land)$ operator, given
$E_1' = \mathcal{S}_{\sigma_c}(E_1)$ and $E_2' = \mathcal{S}_{\sigma_c}(E_2)$, we define
\[
\mathcal{S}_{\sigma_c}(E_1 \land E_2) \;=\;
\begin{cases}
false       & \text{if } E_1' = false \text{ or } E_2' = false,\\[4pt]
E_2'                 & \text{if } E_1' = true,\\[4pt]
E_1'                 & \text{if } E_2' = true,\\[4pt]
E_1' \land E_2'      & \text{otherwise.}
\end{cases}
\]

For $(\lor)$, we use the corresponding short-circuit identities (i.e., $E \lor true \equiv true$ and $E \lor false \equiv E$); otherwise we return $E_1' \lor E_2'$. For $(\oplus)$, we apply constant folding using the identities $E \oplus false \equiv E$ and $E \oplus true \equiv \neg E$, and we omit the full case analysis for brevity.

\subsection{Overall procedure}
\label{subsec:overall}
\SetKwFunction{Simplify}{Simplify}
We now present the overall optimization scheme proposed in this paper. Let $D=(q,c)$ be a dynamic circuit with $|q|=n$ qubits and $|c|=m$ classical bits. The procedure consists of two phases. In Phase~\textit{I}, the input circuit $D$ is rewritten into a semantically equivalent probabilistic circuit $D'$ in which dynamic components are removed or replaced by probabilistic constructs (\cref{subsub:phase1}). In Phase~\textit{II}, executable circuits $D''$ are obtained by sampling and resolving all probabilistic components in $D'$, yielding a concrete instance to run on a backend (\cref{subsub:phase2}).

\subsubsection{Phase~\textit{I}: from dynamic to probabilistic circuits}
\label{subsub:phase1}
Phase~\textit{I} constitutes the core of our optimization framework. In this phase, the input dynamic circuit $D$ is symbolically simulated at compile time on a classical computer in order to infer both quantum and classical information that can be exploited for simplification. The goal is to determine which dynamic components, in particular mid-circuit measurements, resets, and classically controlled operations, can be removed altogether, and which ones can be rewritten into constructs of PCM. The procedure for Phase~\textit{I} is summarized in \cref{alg:phase1}.

Phase~\textit{I} takes as input a dynamic circuit $D=(q,c)$ with $|q|=n$ qubits and $|c|=m$ classical bits, together with two abstract states: $\sigma_q$, which summarizes compile-time information about the quantum register $q$, and $\sigma_c$, which summarizes compile-time information about the classical register $c$. Both $\sigma_q$ and $\sigma_c$ are initialized to match the circuit initialization. The algorithm then scans $D$ instruction-by-instruction, updating $\sigma_q$ via QCP~\cite{qcp} and $\sigma_c$ via CCP (\cref{subsec:ccp}), thereby tracking quantum and classical information throughout the circuit. 
\begin{algorithm}
\SetInd{0.05em}{0.5em}
    \caption{Phase~\textit{I} algorithm (\cref{subsub:phase1}).}
    \label{alg:phase1}
    \KwData{A dynamic circuit $D=(q,c)$, with $|q|=n$ and $|c|=m$. $\sigma_q$ containing information on the state of $q$, and $\sigma_c$ containing information on the state of $c$.}
    \KwResult{A probabilistic circuit $D'$ obtained from $D$}

    \ForEach{$inst \in D.instructions$}{
    \Switch{$inst$}{
        \uCase{$inst$ is a unitary operation}{
        \tcp{Update $\sigma_q$ with QCP \cite{qcp} and append $inst$ to $D'$}
            $\sigma_q \gets \sigma_q.\texttt{update}(inst)$ \\
            $D'.\texttt{append}(inst)$
        }
        \uCase{$inst \equiv \mathrm{measure}(q_j, c_i)$}{
            \uIf{$\sigma_q(j) = \top$} {
                $D'.\texttt{append}(\mathrm{measure}(q_j, c_i))$ \tcp{The meas. is kept in $D'$}
                $\sigma_c \gets \sigma_c.\texttt{update}(\mathrm{measure}(q_j, c_i))$ \tcp{\cref{alg:ccp}}
            }
            \uElseIf{$\sigma_q(j) \in \{\ket{0},\ket{1}\}$}{
            \tcp{Outcome deterministic; update $\sigma_c$ and drop the meas.}
            $\sigma_c \gets \sigma_c.\texttt{update}(\mathrm{measure}(q_j, c_i))$ \tcp{\cref{alg:ccp}}
            }
            \uElse{
                $G \gets \sigma_q.\texttt{group}(q_j)$ \tcp{qubits in entanglement group of $q_j$}
                $D'.\texttt{append}(R_{\sigma_q(j) \rightarrow \ket{0}^{\otimes |G|}}(G))$\\
                $D'.\texttt{append}(\mathfrak{P}[\dots](q_j, \dots) \rightarrow (\mathfrak{C}_i, q_j))$\\
                $\sigma_c \gets \sigma_c.\texttt{update}(\mathrm{measure}(q_j, c_i))$ \tcp{\cref{alg:ccp}}
            $\sigma_q.\texttt{set\_top}(G)$ \tcp{Set to $\top$ all qubits in $G$}
            }
        }
    \uCase{$inst \equiv \mathrm{reset}(q_j)$}{
        \uIf{$\sigma_q(j)=\ket{0}$}{
            \textbf{skip} \tcp{$q_j$ already in $\ket{0}$}
        }
        \uElseIf{$\sigma_q(j)=\ket{1}$}{
            $D'.\texttt{append}(X(q_j))$ \tcp{Apply $X(\ket{1}) \rightarrow \ket{0}$}
        }
        \uElseIf{$\sigma_q(j)=\top$}{
            $D'.\texttt{append}(\mathrm{reset}(q_j))$ \tcp{The reset is kept}
        }
        \uElse {
        $G \gets \sigma_q.\texttt{group}(q_j)$ \tcp{qubits in entanglement group of $q_j$}
                $D'.\texttt{append}(R_{\sigma_q(j) \rightarrow \ket{0}^{\otimes |G|}}(G))$\\
            $D'.\texttt{append}(\mathfrak{P}[\dots](q_j, \dots))$\\
            $\sigma_q.\texttt{set\_top}(G)$ \tcp{Set to $\top$ all qubits in $G$}
        }
        $\sigma_q.\texttt{initialize}(q_j)$ \tcp{Set $\sigma_q(j) \gets \ket{0}$}
    }
    \uCase{$inst \equiv \mathrm{if} \; \texttt{cond} \;\mathrm{then} \; \texttt{B}_{then} \; \mathrm{else} \; \texttt{B}_{else}$}{
        $\texttt{cond}' \gets \mathcal{S}_{\sigma_c}(\texttt{cond})$ \tcp{Simplify $\texttt{cond}$}
        \uIf{$\texttt{cond}' = true$}{
            $D'.\texttt{append}(\texttt{B}_{then})$ \tcp{Only true branch appended}
            $\sigma_q \gets \sigma_q.\texttt{update}(\texttt{B}_{then})$
        }
        \uElseIf{$\texttt{cond}' = false$}{
            $D'.\texttt{append}(\texttt{B}_{else})$ \tcp{Only false branch appended}
            $\sigma_q \gets \sigma_q.\texttt{update}(\texttt{B}_{else})$
        }
        \uElse{
            \tcp{Append operation controlled by $\texttt{cond}'$}
            $D'.\texttt{append}(\mathrm{if} \, \texttt{cond}' \,\mathrm{then} \, \texttt{B}_{then} \, \mathrm{else} \, \texttt{B}_{else})$\\
            $\sigma_q.\texttt{set\_top}(\texttt{B}_{then}, \texttt{B}_{else})$
        }
    }
    }
}
\end{algorithm}
The algorithm iterates over the instruction sequence $D.instructions$. Each instruction $inst$ in $D.instructions$ falls into one of four categories: unitary operations (\cref{par:unit_op}), and three dynamic (non-unitary) constructs, namely, mid-circuit measurements (\cref{par:mcm}), reset operations (\cref{par:resets}), and classically controlled operations (\cref{par:ccop}).
\begin{remark}
We write $\sigma_q(j)$ for the abstract quantum state associated with qubit $q_j$. If $q_j$ is not entangled with any other qubit, then $\sigma_q(j)$ denotes the single-qubit abstract state of $q_j$. Otherwise, when $q_j$ belongs to a non-trivial entanglement group, $\sigma_q(j)$ denotes the abstract state of the entanglement group containing $q_j$.
\end{remark}

\paragraph{Unitary operations}
\label{par:unit_op}
A unitary operation acts only on the quantum register $q$ and therefore leaves the classical abstract state $\sigma_c$ unchanged. We update the quantum abstract state $\sigma_q$ using QCP and append $inst$ to the output circuit $D'$.
\paragraph{Mid-circuit measurements}
\label{par:mcm}
The handling of a measurement instruction $inst \equiv \mathrm{measure}(q_j,c_i)$ depends on the abstract value $\sigma_q(j)$. If $\sigma_q(j)\in\{\ket{0},\ket{1}\}$, measuring $q_j$ in the computational basis leaves its state unchanged and writes a deterministic outcome to $c_i$. In this case, we drop the measurement from $D'$, keep $\sigma_q(j)$ unchanged, and let CCP update $\sigma_c(i)$ to $\ZERO$ or $\ONE$ accordingly, to indicate that the control value is statically determined. If instead $\sigma_q(j)=\top$, no useful quantum information is available about $q_j$ at compile time. The measurement cannot be simplified and is appended unchanged to $D'$, while CCP sets $\sigma_c(i)=\top$ to reflect the unknown value stored in $c_i$. Finally, consider the case where $\sigma_q(j)$ represents a superposition state (i.e., $\sigma_q(j)\notin\{\ket{0},\ket{1},\top\}$). Let $G=\sigma_q.\texttt{group}(q_j)$ be the set of qubits in the entanglement group of $q_j$; if $q_j$ is not entangled with any other qubit, then $G=\{q_j\}$. We remove $\mathrm{measure}(q_j,c_i)$ from $D'$ and replace it with a rotation unitary followed by a probabilistic gate-control binding $\mathfrak{P}[\dots](\dots)\rightarrow(\mathfrak{C}_i,q_j)$, as defined in \cref{def:gate-contr-bind}. After this rewrite, we discard quantum information for the entire group $G$ by setting $\sigma_q(k)$ to $\top$ for all $q_k \in G$, since the post-measurement state depends on the probabilistic instantiation. At the same time, CCP sets $\sigma_c(i)=\PROB$, indicating that the measurement outcome is represented by the probabilistic control $\mathfrak{C}_i$ through the gate-control binding.
\paragraph{Resets}
\label{par:resets}
When $inst \equiv \mathrm{reset}(q_j)$, its treatment is determined by the current abstract value $\sigma_q(j)$. If $\sigma_q(j)=\ket{0}$, the reset is redundant and is therefore omitted from $D'$, leaving $\sigma_q(j)$ unchanged. If $\sigma_q(j)=\ket{1}$, we replace $\mathrm{reset}(q_j)$ in $D'$ with a single-qubit unitary (e.g., an $X$ gate) that maps $\ket{1}$ to $\ket{0}$, thus reinitializing $q_j$. If $\sigma_q(j)=\top$, the reset cannot be rewritten and is copied unchanged to $D'$. When $q_j$ is in a superposition (i.e., $\sigma_q(j)\notin\{\ket{0},\ket{1},\top\}$), let $G=\sigma_q.\texttt{group}(q_j)$ be the set of qubits in the entanglement group of $q_j$. We apply the construction of \cref{subsec:ref_rep_mcmres}: we replace the reset with a rotation unitary acting on the whole group $G$ that maps its current state to $\ket{0}^{\otimes |G|}$, followed by a probabilistic gate whose branches all act as the identity on $q_j$, ensuring that $q_j$ ends in $\ket{0}$ in every instantiation.
The abstract state $\sigma_q$ is updated by conservatively discarding quantum information for the group $G$, i.e., setting $\sigma_q(k)=\top$ for all $q_k \in G$, since the post-reset state of the group depends on the probabilistic instantiation. In all cases, $\sigma_q(j)$ is finally set to $\ket{0}$ to reflect that $q_j$ is reinitialized by the reset.

\paragraph{Classically Controlled Operations}
\label{par:ccop}
When $inst$ is a classically controlled
operation $\mathrm{if} \; \texttt{cond} \;\mathrm{then} \; \texttt{B}_{then} \; \mathrm{else} \; \texttt{B}_{else}$ (\cref{fig:cls_ctrl_op}),
we use the simplification function $\mathcal{S}_{\sigma_{c}}$ defined in \cref{subsec:ccop-simp} that exploits the abstract state $\sigma_c$ to transform the condition $\texttt{cond}$ into a possibly simplified probabilistic condition $\texttt{cond}'$.
If $\mathcal{S}_{\sigma_c}(\texttt{cond})$ evaluates to $true$ or $false$, the condition can be eliminated entirely and the operation reduces to a single unconditional branch. In this case, we append the selected branch ($\texttt{B}_{then}$ if $true$, and $\texttt{B}_{else}$ if $false$) to $D'$ and update the quantum abstract state $\sigma_q$ by simulating the unitary operations in the appended branch with QCP.
Otherwise, $\mathcal{S}_{\sigma_c}(\texttt{cond})$ returns a probabilistic condition $\texttt{cond}' \notin \{true,false\}$. In this case, we append to $D'$ the classically controlled operation
$\mathrm{if}\ \texttt{cond}'\ \mathrm{then}\ \texttt{B}_{then}\ \mathrm{else}\ \texttt{B}_{else}$.
Since the executed branch depends on runtime information, we conservatively update $\sigma_q$ by discarding quantum information for all qubits affected by the operations in $\texttt{B}_{then}$ or $\texttt{B}_{else}$.
The condition $\texttt{cond}'$ may include probabilistic controls $\mathfrak{C}_i$, which are resolved in Phase~\textit{II} when probabilistic components are instantiated. Conversely, any parts of $\texttt{cond}$ that are statically determined at compile time are folded away in $\texttt{cond}'$.

\subsubsection{Phase \textit{II}: from probabilistic circuits to executable instances}
\label{subsub:phase2}
Phase~\textit{II} is responsible for turning the probabilistic circuit $D'$ into executable circuits. Concretely, it generates executable instances $D''$ by traversing $D'$ instructions and resolving every probabilistic component. The procedure for Phase~\textit{II} is summarized in \cref{alg:phase2}.
\begin{algorithm}
\SetInd{0.05em}{0.5em}
\caption{Phase~\textit{II} algorithm (\cref{subsub:phase2}).}
\label{alg:phase2}
\KwData{A probabilistic circuit $D' = (q, c)$ with $|q|=n$, $|c|=m$, and probabilistic controls $\mathfrak{C}_i$ for $1 \leq i \leq m$.}
\KwResult{An executable instance $D''$ of $D'$}
\ForEach{$inst \in D'.instructions$}{
    \Switch{$inst$}{
        \tcp{$l$-qubit probabilistic gate for resets (\cref{eq:probgate_for_res})}
        \uCase{$\mathfrak{P}[(I \otimes \mathcal{R}_{\ket{0}}, p_1), (I \otimes \mathcal{R}_{\ket{1}}, p_2)](q_j, q_{r_1}, \dots, q_{r_{l-1}})$}{
            $\texttt{choices} \gets [I \otimes \mathcal{R}_{\ket{0}}, I \otimes \mathcal{R}_{\ket{1}}]$ \\
            $\texttt{prob\_distr} \gets [p_1, \, p_2]$ \\
            $\texttt{choice} \gets \texttt{rand}(\texttt{choices}, \,\texttt{prob\_distr})$ \\
            $D''.\texttt{append}(\texttt{choice}(q_j, q_{r_1}, \dots, q_{r_{l-1}}))$
        }
        \tcp{$l$-qubit probabilistic gate-control binding (\cref{def:gate-contr-bind})}
        \uCase{$\mathfrak{P}[(I \otimes \mathcal{R}_{\ket{0}}, p_1), (X \otimes \mathcal{R}_{\ket{1}}, p_2)](q_j,q_{r_1},\dots, q_{r_{l-1}})
    \rightarrow
    (\mathfrak{C}_{i}, q_j)$}{
            $\texttt{choices} \gets [I \otimes \mathcal{R}_{\ket{0}}, X \otimes \mathcal{R}_{\ket{1}}]$ \\
            $\texttt{prob\_distr} \gets [p_1, p_2]$ \\
            $\texttt{choice} \gets \texttt{rand}(\texttt{choices}, \texttt{prob\_distr})$ \\
            $D''.\texttt{append}(\texttt{choice}(q_j, q_{r_1}, \dots, q_{r_{l-1}}))$ \\
            $\mathfrak{C}_i \gets \mathfrak{0} \;\textbf{if} \; \texttt{choice} = I \otimes \mathcal{R}_{\ket{0}} \; \textbf{else} \; \mathfrak{1}$
        }
        \uCase{$inst \equiv \mathrm{if} \; \texttt{cond}' \;\mathrm{then} \; \texttt{B}_{then} \; \mathrm{else} \; \texttt{B}_{else}$}{
            $\texttt{cond}'' \gets \mathcal{S}_{\mathfrak{C}}(\texttt{cond}')$ \tcp{Simplify $\texttt{cond}'$ (\cref{eq:simp-fun-prob-cond})}
        \uIf{$\texttt{cond}'' = true$}{
            $D''.\texttt{append}(\texttt{B}_{then})$ \tcp{Only true branch appended with no condition}
        }
        \uElseIf{$\texttt{cond}'' = false$}{
            $D''.\texttt{append}(\texttt{B}_{else})$ \tcp{Only false branch appended }
        }
        \uElse{
            \tcp{Append operation controlled by $\texttt{cond}''$}
            $D''.\texttt{append}(\mathrm{if} \, \texttt{cond}'' \,\mathrm{then} \, \texttt{B}_{then} \, \mathrm{else} \, \texttt{B}_{else})$
        }
        }
        \uCase{otherwise}{
            $D''.\texttt{append}(inst)$
        }
    }
}
\end{algorithm}
Unlike Phase~\textit{I}, Phase~\textit{II} performs no circuit simulation. During the traversal of $D'$, the only state we maintain is the current assignment of probabilistic controls. For the executable instance being generated, each probabilistic control $\mathfrak{C}_i$ takes a value in $\{\mathfrak{0},\mathfrak{1}\}$, and this value is set whenever the corresponding probabilistic gate-control binding is resolved, i.e., when the associated probabilistic gate is instantiated.
These assignments are then used to resolve probabilistic conditions. 
To this end, we define a simplification function
\begin{equation}
\label{eq:simp-fun-prob-cond}
\mathcal{S}_{\mathfrak{C}} : \mathcal{E}' \to \{true,false\} \cup \mathcal{E},
\end{equation}
where $\mathcal{E}'$ is the set of probabilistic conditions (\cref{def:prob-cond}) and $\mathcal{E}$ is the set of Boolean expressions over classical bits only. Given a probabilistic guard $\texttt{cond}'$ and the current assignment of probabilistic controls, $\mathcal{S}_{\mathfrak{C}}$ first replaces each occurrence of $\mathfrak{C}_i$ with $false$ if $\mathfrak{C}_i=\mathfrak{0}$ and with $true$ if $\mathfrak{C}_i=\mathfrak{1}$. It then simplifies the resulting expression by short-circuit identities and constant folding as in $\mathcal{S}_{\sigma_c}$ (\cref{subsec:ccop-simp}), returning either $true$ or $false$ (thus statically evaluating the guard at this point) or a residual condition in $\mathcal{E}$, which is preserved in $D''$.

Phase~\textit{II} explicitly processes only instructions $inst$ in $D'.instructions$ involving probabilistic constructs, namely probabilistic gates (\cref{par:prob-gate}), probabilistic gate-control bindings (\cref{par:prob-gate-ctrl-bind}), and operations controlled by probabilistic conditions (\cref{par:ccop-ph2}).
\paragraph{Probabilistic gates}
\label{par:prob-gate}
If $inst$ is a probabilistic gate introduced by the reset rewriting,
$\mathfrak{P}[(I \otimes \mathcal{R}_{\ket{0}}, p_1), (I \otimes \mathcal{R}_{\ket{1}}, p_2)](q_j,q_{r_1},\dots, q_{r_{l-1}})$,
Phase~\textit{II} resolves it by sampling one of its two alternatives according to the distribution $(p_1,p_2)$. The sampled unitary, either $I \otimes \mathcal{R}_{\ket{0}}$ or $I \otimes \mathcal{R}_{\ket{1}}$, is then applied to $(q_j,q_{r_1},\dots, q_{r_{l-1}})$ and appended to $D''$.

\paragraph{Probabilistic gate-control bindings}
\label{par:prob-gate-ctrl-bind}
If $inst$ is a probabilistic gate-control binding introduced by the mid-circuit measurement rewriting, $
\mathfrak{P}[(I \otimes \mathcal{R}_{\ket{0}}, p_1), (X \otimes \mathcal{R}_{\ket{1}}, p_2)]
(q_j, q_{r_1}, \dots, q_{r_{l-1}})
\rightarrow (\mathfrak{C}_i, q_j)$,
then we sample one of the two unitary alternatives of $\mathfrak{P}$ according to $(p_1,p_2)$ and append the sampled operation to the instance $D''$.
At the same time, we update the associated probabilistic control $\mathfrak{C}_i$ to $\mathfrak{0}$ if $I \otimes \mathcal{R}_{\ket{0}}$ is chosen, and to $\mathfrak{1}$ if $X \otimes \mathcal{R}_{\ket{1}}$ is chosen, as in~\cref{eq:prob-ctrl-update}.
\paragraph{Classically controlled operations}
\label{par:ccop-ph2}
In this case, $inst~\equiv~ \mathrm{if} \; \texttt{cond}' \;\mathrm{then} \; \texttt{B}_{then} \; \mathrm{else} \; \texttt{B}_{else}$,
with a probabilistic condition $\texttt{cond}'$. We simplify $\texttt{cond}'$ with $\mathcal{S}_{\mathfrak{C}}$ (\cref{eq:simp-fun-prob-cond}) under the current assignment of probabilistic controls. 
If $\mathcal{S}_{\mathfrak{C}}(\texttt{cond}')$ evaluates to $true$ or $false$, we unconditionally append to $D''$ only the branch $\texttt{B}_{then}$ or $\texttt{B}_{else}$, respectively.
Otherwise, $\mathcal{S}_{\mathfrak{C}}(\texttt{cond}')$ returns a residual Boolean expression $\texttt{cond}''$ over classical controls, and we append to $D''$ the guarded operation $\mathrm{if} \; \texttt{cond}'' \;\mathrm{then} \; \texttt{B}_{then} \; \mathrm{else} \; \texttt{B}_{else}$.

Instructions that do not involve probabilistic constructs, such as
unitary gates, resets, and measurements that were left unchanged in $D'$ after
Phase~\textit{I}, are simply copied to $D''$.

At the end of Phase~\textit{II}, all probabilistic gates in $D'$ have been instantiated and all probabilistic conditions have been simplified as far as permitted by the sampled values of the probabilistic controls $\mathfrak{C}_i$, yielding an executable circuit instance $D''$ with no remaining probabilistic components. Re-running Phase~\textit{II} on the same $D'$ generally produces different executable instances $D''$.
\begin{remark}[Runtime analysis of the method]
The runtime cost of the proposed compilation pipeline is dominated by Phase~\textit{I}. In particular, the dominant costs arise from the propagation of quantum information and from the synthesis of the unitary rotations used to replace measurements and resets with probabilistic constructs. In the worst case, these components can be exponential in the number of qubits.
In practice, however, for quantum-state propagation, we use QCP that tracks entanglement groups only up to a fixed size threshold, as discussed in \cref{subsec:qcp}. This threshold also limits the cost of unitary synthesis, since it bounds the size of the states that may need to be synthesized. Moreover, as discussed in \cref{subsec:comp-cost-rewriting}, one may additionally follow the approach of~\cite{big_prob} and impose a threshold on the size of the states to be synthesized, above which a measurement or reset is left unchanged instead of being rewritten.
A detailed runtime analysis of these components is provided in~\cite{big_prob, qcp}.
Phase~\textit{II} only instantiates probabilistic constructs by scanning the instruction sequence and therefore adds linear overhead that does not affect the overall complexity.

\end{remark}

\subsection{Correctness of the method}
\label{sec:method-correctness}
Our compilation pipeline preserves the semantics of the input dynamic circuit $D$ in terms of the probability distribution over its execution paths. A single execution of $D$ follows one execution path, and the possible execution paths, together with their probabilities, are determined by the non-unitary operations.
In particular, mid-circuit measurements, resets, and classically controlled operations introduce branching in the evolution of the quantum state.
Each such branching point gives rise to multiple execution paths.
Our method preserves all execution paths together with their associated probabilities: Phase~\textit{I} constructs a probabilistic circuit $D'$ that represents exactly the execution paths of $D$, while Phase~\textit{II} generates executable instances from $D'$ according to the same probability distribution.

\begin{theorem}
\label{theorem}
Let $D'$ be the probabilistic circuit produced by Phase~\textit{I} from a dynamic circuit $D$. Then $D'$ represents exactly the execution paths of $D$, preserving the corresponding quantum states and their associated probabilities.
\end{theorem}
\begin{proof}
   The proof proceeds by induction on the circuit instructions, showing that each rewrite rule used in Phase~\textit{I} is semantics-preserving by construction.
    \begin{itemize}
        \item \textit{Resets.} A reset of a qubit $q_j$ is replaced only when the abstract state provides enough information about $q_j$. If $q_j$ is already in state $\ket{0}$, the reset is omitted, while if $q_j$ is in state $\ket{1}$, it is replaced by a gate mapping $\ket{1}$ to $\ket{0}$, preserving the same effect as the reset. Otherwise, the reset is replaced with a probabilistic gate (\cref{subsub:patt_res_emu}) that encodes the same post-reset states, with the same probabilities, as in the original dynamic circuit. Hence, the rewritten fragment preserves the same execution paths, together with the resulting quantum states.
        \item \textit{Measurements.}
        Consider the operation $\mathrm{measure}(q_j,c_i)$. If $q_j$ is in state $\ket{0}$ or $\ket{1}$, the measurement does not alter the quantum state and only writes $0$ or $1$, respectively, to $c_i$. In this case, the measurement is omitted, and CCP updates the abstract state of $c_i$ accordingly. If no information about $q_j$ is available, the measurement is kept unchanged and CCP marks $c_i$ as unknown. Otherwise, the measurement is replaced by a probabilistic gate-control binding, which encodes the possible outcomes with the corresponding probabilities and post-measurement states. Moreover, CCP updates the abstract classical state to reflect that $c_i$ is replaced with its associated probabilistic control, whose value is determined by the probabilistic gate-control binding. Hence, the rewritten fragment preserves the same execution paths, together with the corresponding post-measurement states, while CCP consistently tracks the measurement outcomes.
        \item \textit{Classically controlled operations.}
        The simplification function $\mathcal{S}_{\sigma_c}$ (\cref{eq:simp-fun-cond}) uses the abstract classical state maintained by CCP, which is consistent with the measurement outcomes as discussed above. Therefore, if $\mathcal{S}_{\sigma_c}$ evaluates to $true$ or $false$, the selected branch is exactly the one taken at runtime in the original circuit, and the classical condition can be removed without changing the execution path.
        Otherwise, $\mathcal{S}_{\sigma_c}$ returns a residual guard that contains the part of the condition that cannot be resolved at compile time: classical bits whose value remains unknown are kept, while bits marked as probabilistic by CCP are replaced by their corresponding probabilistic controls. All parts of the guard whose value is already determined are folded away by $\mathcal{S}_{\sigma_c}$.
        Thus, the rewritten classically controlled operation preserves the same behavior.
    \end{itemize}
Since each rewriting maintains the same post-operation quantum states of the replaced non-unitary operation, together with their corresponding probabilities,  the execution paths of $D$ are preserved in $D'$ with the same probability distribution.
\end{proof}
By \Cref{theorem}, Phase~\textit{I} builds a probabilistic circuit $D'$ that represents the execution paths of $D$. Phase~\textit{II} then resolves the probabilistic components introduced in Phase~\textit{I}.
More precisely, probabilistic gates and probabilistic gate-control bindings are sampled according to the same probability distributions as the corresponding rewritten resets and measurements, and probabilistic controls are resolved consistently with their associated probabilistic gate-control bindings. As a result, each executable instance $D''$ produced in Phase~\textit{II} corresponds to one execution path of $D$, with the corresponding quantum state and probability. Hence, the overall compilation pipeline preserves the semantics of $D$.

\section{Evaluation}\label{sec:evaluation}
The goal of this evaluation is to assess how effectively the proposed method simplifies dynamic circuits. We benchmark it on a large dataset of randomly generated dynamic circuits (\cref{subsec:rand-circ-eval}) and complement these results with a case study based on a dynamic circuit from the literature (\cref{subsec:dem-examp}).

We implemented our framework in Python using Qiskit as the underlying representation for quantum circuits \cite{ibmOverviewQuantum, ibmLibrarylatest}.
The source code of the framework is publicly available at \url{https://github.com/i2-tum/pcm-ccop-dc-pass}.
The framework takes as input a dynamic circuit represented as a Qiskit \texttt{QuantumCircuit} object and produces a probabilistic intermediate circuit $D'$ (Phase~\textit{I}) as well as executable instances $D''$ (Phase~\textit{II}), returned as \texttt{QuantumCircuit} objects, obtained by resolving probabilistic constructs.

\subsection{Evaluation on random circuits}
\label{subsec:rand-circ-eval}
\begin{figure}[t]
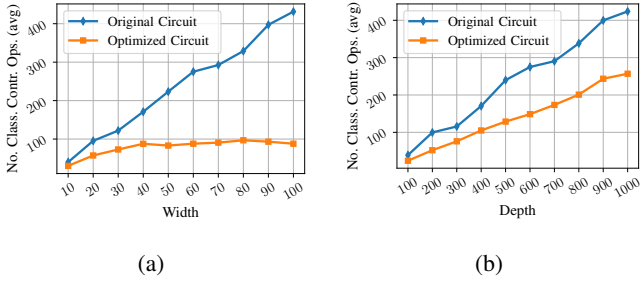

  \centering
  \subfloat[]{%
    \resizebox{0.495\linewidth}{!}{\input{images/ctrl_ops_vs_nqubits.pgf}}%
    \label{subfig:ctrl_ops_vs_nqubits}%
  }\hfill
  \subfloat[]{%
    \resizebox{0.495\linewidth}{!}{\input{images/ctrl_ops_vs_depth.pgf}}%
    \label{subfig:ctrl_ops_vs_depth}%
  }

  \caption{Number of classically controlled operations in the original circuits and in the optimized probabilistic circuits produced by Phase~\textit{I}. Values are averages over $10$ random circuits per setting. \cref{subfig:ctrl_ops_vs_nqubits} varies the number of qubits (width) from $w=10$ to $w=100$ with fixed depth $d=250$, while \cref{subfig:ctrl_ops_vs_depth} varies the depth from $d=100$ to $d=1000$ (step $100$) with fixed width $w=25$.}
  \label{fig:exper}
\end{figure}
\subsubsection{Benchmarks}
\label{subsubsec:bench}
To obtain a large and structurally diverse benchmark, we generate random dynamic circuits by extending Qiskit's \texttt{random\_circuit} utility from the \texttt{qiskit.circuit.random} module, which produces static circuits of fixed width and depth \cite{ibmRandomlatest}. Our extended random circuit generator allows to inject mid-circuit measurements, creating intermediate bit values, classically controlled operations
with guards drawn from the three syntactic patterns depicted in \cref{subsec:dc_ccop} that use these measurement outcomes, and resets to reinitialize the measured qubits. 
We generate $10$ circuits for each configuration, varying the number of qubits from $10$ to $100$ and depth from $100$ to $1000$.
\subsubsection{Metric}
We measure the impact of the optimization by counting how many classically controlled operations become independent of runtime classical information. Concretely, we count a classically controlled operation as removed whenever Phase~\textit{I} rewrites its guard $\texttt{cond}$ into $\texttt{cond}'=\mathcal{S}_{\sigma_c}(\texttt{cond})$ such that $\texttt{cond}'$ contains no classical bits, i.e., $\texttt{cond}'\in\{true,false\}$ or $\texttt{cond}'$ is a Boolean expression over probabilistic controls $\mathfrak{C}_i$ only. In this case, the corresponding executable instances $D''$ require no runtime classical feedforward, and the guarded construct is resolved into an unconditional execution of exactly one branch.
This metric is conservative: we only count removals that are guaranteed after Phase~\textit{I}. In Phase~\textit{II}, hybrid guards that still contain classical bits and probabilistic controls (e.g., $c_1 \land \mathfrak{C}_2$) may simplify further for specific sampled instances (e.g., if $\mathfrak{C}_2=\mathfrak{0}$ then the guard becomes $false$). We do not perform instance-by-instance analysis, and report only this conservative removal count.
\subsubsection{Results}
\label{subsub:exp-res}
Figure~\ref{fig:exper} summarizes the results where each data point reports the mean number of classically controlled operations over $10$ different circuits for the same setting.

In \cref{subfig:ctrl_ops_vs_nqubits} we fix the depth to $d=250$ and vary the circuit width from $w=10$ to $w=100$. We observe that the fraction of removed operations (i.e., the gap between the two curves) increases with $w$. This trend is explained by how much quantum information can be retained during Phase~\textit{I}. Whenever a measurement is rewritten into a probabilistic gate-control binding, \cref{alg:phase1} conservatively sets the abstract state of all qubits in the corresponding entanglement group to $\top$. Once a qubit is mapped to $\top$, subsequent measurements on that qubit cannot be emulated via probabilistic bindings, limiting the propagation of classical information needed to simplify later conditions. At fixed depth, increasing $w$ spreads the sampled operations over a larger register: on average, each qubit participates in fewer interactions, and the same subset of qubits is less likely to interact repeatedly. This reduces the formation of large entanglement groups and reduces the propagation of $\top$, thereby increasing the removal rate.

In \cref{subfig:ctrl_ops_vs_depth}, we vary the depth from $100$ to $1000$ in steps of $100$ while keeping the width fixed to $w=25$. Unlike \cref{subfig:ctrl_ops_vs_nqubits}, the fraction of removed operations remains comparatively stable and slightly decreases as depth increases. At fixed width, deeper circuits tend to create larger and more persistent entanglement groups, which limits the propagation of classical information for condition simplification. Nevertheless, the removal rate does not collapse at large depths: reset operations periodically reinitialize qubits, and thus effectively restart information propagation, enabling further condition simplifications even in very deep circuits.

From \cref{fig:exper}, it can be observed that Phase~\textit{I} removes a substantial fraction of classically controlled operations, thereby reducing the amount of runtime classical feedforward. In \cref{subfig:ctrl_ops_vs_nqubits} (fixed depth $d=250$), for $w \ge 80$ the number of controlled operations drops from roughly $380$--$450$ to about $90$--$100$, corresponding to an elimination of around $75\%$--$80\%$. In \cref{subfig:ctrl_ops_vs_depth} (fixed width $w=25$), the removal rate is more stable and stays in the $35\%$--$50\%$ range. Overall, these reductions can translate into lower execution latency by decreasing the overhead for classical feedforward. While this shifts additional compilation work to the classical computer, it yields simpler quantum circuits to be run on quantum devices, which may be desirable in practice since quantum resources are currently more limited than classical ones.

\subsection{Demonstrative example}
\label{subsec:dem-examp}
We demonstrate the effect of our optimization pass on a dynamic-circuit algorithm from the literature for preparing GHZ states \cite{ghz}.
The example illustrates how our two-phase procedure removes classical feedforward when sufficient compile-time information is available, while conservatively preserving dynamic components otherwise.

We consider a $5$-qubit dynamic circuit for GHZ-state preparation on qubits $q_1, q_2, q_3, q_4, q_5$. The circuit starts with a unitary prefix that prepares a uniform superposition over the $8$ computational-basis assignments of $q_1,q_3,q_5$, while encoding the corresponding parity information into $q_2$ and $q_4$. This superposition over the qubits $q_1, q_2, q_3, q_4, q_5$ is the state:
\begin{equation*}
\begin{split}
\ket{\psi}
&=
\frac{1}{\sqrt{8}}\sum_{x_1,\, x_3,\, x_5 \,\in\,\{0,1\}}
\\
&\quad
\ket{x_1}_{q_1}\ket{x_1\oplus x_3}_{q_2}\ket{x_3}_{q_3}
\ket{x_3\oplus x_5}_{q_4}\ket{x_5}_{q_5}.
\end{split}
\end{equation*}
The dynamic circuit $D = (q, c)$ is defined as:

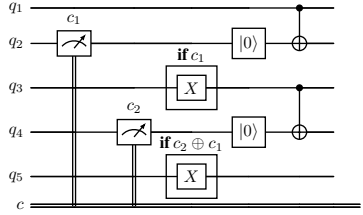
\begin{center}
\resizebox{0.55\linewidth}{!}{%
\begin{quantikz}[row sep=0.35cm, column sep=0.55cm, wire types={q,q,q,q,q,c,c}]
  \lstick{$q_1$} & \qw & \qw & \qw & \qw & \ctrl{1} & \qw
\\
  \lstick{$q_2$} & \meter{c_1}\wire[d][4]{c} & \qw & \qw & \gate{\mathrm{\ket{0}}} & \targ{} & \qw
\\
  \lstick{$q_3$} & \qw & \qw & \gate{X}\gategroup[
      wires=1,
      steps=1,
      style={draw,inner xsep=3pt,inner ysep=2pt},
      label style={label position=above}
    ]{\textbf{if}$\,c_1$} & \qw & \ctrl{1} & \qw
\\
  \lstick{$q_4$} & \qw & \meter{c_2}\wire[d][2]{c} & \qw & \gate{\mathrm{\ket{0}}} & \targ{} & \qw
\\
  \lstick{$q_5$} & \qw & \qw & \gate{X}\gategroup[
      wires=1,
      steps=1,
      style={draw,inner xsep=3pt,inner ysep=2pt},
      label style={label position=above}
    ]{\textbf{if}$\,c_2 \oplus c_1$} & \qw & \qw & \qw
\\
\lstick{$c$}
  & \cw & \cw & \cw & \cw & \cw & \cw & \cw
\end{quantikz}
}
\end{center}
where the quantum register $q$ is initialized in the joint state $\ket{\psi}$.
Collecting terms with respect to $q_2$ and normalizing, we can write
$
\ket{\psi}
=\frac{1}{\sqrt{2}}
(\ket{0}_{q_2}\ket{\phi_0}_{q_1 q_3 q_4 q_5}
+\ket{1}_{q_2}\ket{\phi_1}_{q_1 q_3 q_4 q_5}),
$
where $\ket{\phi_0}$, $\ket{\phi_1}$ are normalized states on $q_1, q_3, q_4, q_5$, and both are superpositions of four computational basis states.
Consequently, measuring $q_2$ yields outcome $0$ or $1$ with probability $1/2$, and the post-measurement state is
$\ket{0}_{q_2}\ket{\phi_0}$ or $\ket{1}_{q_2}\ket{\phi_1}$, respectively. Applying the Phase~\textit{I} algorithm (\cref{alg:phase1}) to $D$ produces the probabilistic circuit $D'$:
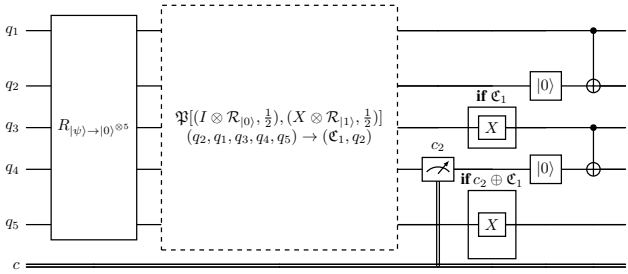
\begin{center}
\resizebox{0.95\linewidth}{!}{%
\begin{quantikz}[row sep=0.35cm, column sep=0.55cm, wire types={q,q,q,q,q,c}]
  \lstick{$q_1$} & \gate[wires=5]{R_{\ket{\psi} \rightarrow \ket{0}^{\otimes 5}}} & \gate[
  wires=5,
  style={dashed}
]{
\begin{array}{c}
\mathfrak{P}[(I \otimes \mathcal{R}_{\ket{0}}, \tfrac{1}{2}),
(X \otimes \mathcal{R}_{\ket{1}}, \tfrac{1}{2})]\\(q_2, q_1, q_3, q_4, q_5)
\rightarrow (\mathfrak{C}_1, q_2)
\end{array}
} & \qw & \qw & \qw & \ctrl{1} & \qw
\\
  \lstick{$q_2$} &                   &                                  & \qw & \qw & \gate{\mathrm{\ket{0}}} & \targ{} & \qw
\\
  \lstick{$q_3$} &                   &                                  & \qw & \gate{X}\gategroup[
      wires=1,
      steps=1,
      style={draw,inner xsep=3pt,inner ysep=2pt},
      label style={label position=above}
    ]{\textbf{if}$\;\mathfrak{C}_1$} & \qw & \ctrl{1} & \qw
\\
  \lstick{$q_4$} &                   &                                  & \meter{c_2}\wire[d][2]{c} & \qw & \gate{\mathrm{\ket{0}}} & \targ{} & \qw
\\
  \lstick{$q_5$} &                   &                                  & \qw & \gate{X}\gategroup[
      wires=1,
      steps=1,
      style={draw,inner xsep=3pt,inner ysep=2pt},
      label style={label position=above}
    ]{\textbf{if}$\,c_2 \oplus \mathfrak{C}_1$} & \qw & \qw & \qw
\\
  \lstick{$c$} & \cw & \cw & \cw & \cw & \cw & \cw & \cw
\end{quantikz}
}
\end{center}
Phase~\textit{I} replaces $D.\mathrm{measure}(q_2, c_1)$ with a probabilistic gate-control binding $D'.\mathfrak{P}[\cdot](q_2,\dots)\rightarrow(\mathfrak{C}_1,q_2)$, shown as the dashed box in the circuit diagram, which captures the two post-measurement branches. Subsequent classical conditions are then rewritten to depend on $\mathfrak{C}_1$ instead of $c_1$. The measurement of $q_4$ is kept unchanged in $D'$, because, after removing the first mid-circuit measurement,
QCP conservatively loses track of the entanglement group containing $q_2$, and hence also $q_4$, which implies $\sigma_q(4)=\top$.

Phase~\textit{II} (\cref{alg:phase2}) samples the probabilistic gate in $D'$, producing an executable instance $D''$ where the sampled branch sets $\mathfrak{C}_1$ to either $\mathfrak{0}$ or $\mathfrak{1}$ and simplifies the remaining conditions. When the probabilistic gate-control binding is instantiated as $(X \otimes \mathcal{R}_{\ket{1}})$, then $\mathfrak{C}_1 \gets \mathfrak{1}$. Consequently, the guard $\mathfrak{C}_1$ becomes $true$, $X(q_3)$ is applied unconditionally, and $c_2 \oplus \mathfrak{C}_1$ simplifies to $\neg c_2$, yielding the executable instance:
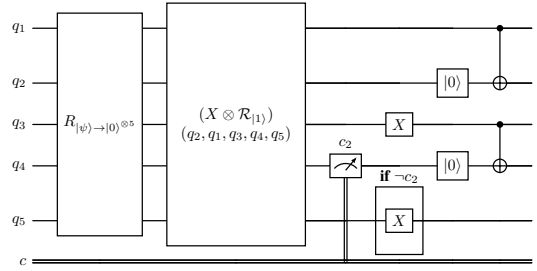
\begin{center}
\resizebox{0.8\linewidth}{!}{%
\begin{quantikz}[row sep=0.35cm, column sep=0.55cm, wire types={q,q,q,q,q,c}]
  \lstick{$q_1$} 
    & \gate[wires=5]{R_{\ket{\psi} \rightarrow \ket{0}^{\otimes 5}}}
    & \gate[wires=5]{
      \begin{array}{c}
      (X \otimes \mathcal{R}_{\ket{1}})\\
      (q_2, q_1, q_3, q_4, q_5)
      \end{array}
    }
    & \qw 
    & \qw 
    & \qw 
    & \ctrl{1} 
    & \qw
\\
  \lstick{$q_2$} 
    & 
    & 
    & \qw 
    & \qw 
    & \gate{\mathrm{\ket{0}}} 
    & \targ{} 
    & \qw
\\
  \lstick{$q_3$} 
    & 
    & 
    & \qw 
    & \gate{X} 
    & \qw 
    & \ctrl{1} 
    & \qw
\\
  \lstick{$q_4$} 
    & 
    & 
    & \meter{c_2}\wire[d][2]{c} 
    & \qw 
    & \gate{\mathrm{\ket{0}}} 
    & \targ{} 
    & \qw
\\
  \lstick{$q_5$} 
    & 
    & 
    & \qw 
    & \gate{X}\gategroup[
        wires=1,
        steps=1,
        style={draw,inner xsep=3pt,inner ysep=2pt},
        label style={label position=above}
      ]{\textbf{if}$\;\neg c_2$}
    & \qw 
    & \qw 
    & \qw
\\
  \lstick{$c$} 
    & \cw & \cw & \cw & \cw & \cw & \cw & \cw
\end{quantikz}
}
\end{center}
On the other hand, when the probabilistic gate-control binding is resolved as $(I \otimes \mathcal{R}_{\ket{0}})$, then $\mathfrak{C}_1 \gets \mathfrak{0}$. The guard $\mathfrak{C}_1$ becomes $false$, $X(q_3)$ can be discarded, and $c_2 \oplus \mathfrak{C}_1$ simplifies to $c_2$.
Thus, the produced executable instance is the circuit:
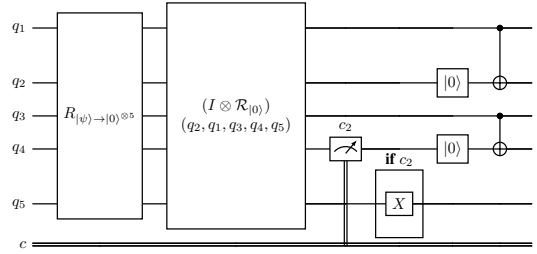
\begin{center}
\resizebox{0.8\linewidth}{!}{%
\begin{quantikz}[row sep=0.35cm, column sep=0.55cm, wire types={q,q,q,q,q,c}]
  \lstick{$q_1$} 
    & \gate[wires=5]{R_{\ket{\psi} \rightarrow \ket{0}^{\otimes 5}}}
    & \gate[wires=5]{
      \begin{array}{c}
      (I \otimes \mathcal{R}_{\ket{0}})\\
      (q_2, q_1, q_3, q_4, q_5)
      \end{array}
    }
    & \qw 
    & \qw 
    & \qw 
    & \ctrl{1} 
    & \qw
\\
  \lstick{$q_2$} 
    & 
    & 
    & \qw 
    & \qw 
    & \gate{\mathrm{\ket{0}}} 
    & \targ{} 
    & \qw
\\
  \lstick{$q_3$} 
    & 
    & 
    & \qw 
    & \qw
    & \qw 
    & \ctrl{1} 
    & \qw
\\
  \lstick{$q_4$} 
    & 
    & 
    & \meter{c_2}\wire[d][2]{c} 
    & \qw 
    & \gate{\mathrm{\ket{0}}} 
    & \targ{} 
    & \qw
\\
  \lstick{$q_5$} 
    & 
    & 
    & \qw 
    & \gate{X}\gategroup[
        wires=1,
        steps=1,
        style={draw,inner xsep=3pt,inner ysep=2pt},
        label style={label position=above}
      ]{\textbf{if}$\;c_2$}
    & \qw 
    & \qw 
    & \qw
\\
  \lstick{$c$} 
    & \cw & \cw & \cw & \cw & \cw & \cw & \cw
\end{quantikz}
}
\end{center}

\section{Conclusion and Future Work}\label{sec:conclusion}
We presented a compile-time optimization framework for reducing classical feedforward in dynamic circuits.
The key idea is to replace part of the runtime classical feedback loop between the QPU and the controller with compile-time rewrites, so that fewer decisions require measurement outcomes to be routed out, processed, and sent back as control signals. As a result, the optimized circuits reduce classical-control overheads and can mitigate latency on hardware. We evaluated the approach on a large dataset of random dynamic circuits and on a representative GHZ-state preparation routine. 
As future work, it would be useful to equip the optimization pass with an explicit, hardware-aware cost model (e.g., mid-circuit measurement overhead and classical feedforward latency) to quantify its benefits on specific target devices.


\section*{Acknowledgments}
The research is part of the Munich Quantum Valley (MQV), which is supported by the Bavarian state government with funds from the Hightech Agenda Bayern Plus.

\bibliographystyle{ieeetr}
\bibliography{references}

\end{document}